\documentclass[11pt,a4paper]{article}

\usepackage{amsmath} 
\usepackage{amsthm}
\usepackage{amssymb}

\usepackage{url}
\usepackage[hidelinks]{hyperref} %hide red (blue,green) boxes around links
\usepackage{color}
\usepackage{xcolor} 
\definecolor{mauve}{rgb}{0.58,0,0.82}

\let\chapter\section
\usepackage[ruled, vlined, linesnumbered]{algorithm2e}

\SetKw{True}{true}
\SetKw{False}{false}
\SetKwData{typeInt}{Int}
\SetKwData{typeRat}{Rat}
\SetKwData{Defined}{Defined}
\SetKwFunction{parseStatement}{parseStatement}

\usepackage{verbatim}

\usepackage{proof}
\theoremstyle{definition}
\newtheorem{theorem}{Theorem}
\newtheorem{lemma}[theorem]{Lemma}
\newtheorem{example}{Example}
\newtheorem{remark}{Remark}
\newtheorem{corollary}{Corollary}
\newtheorem{definition}{Definition}
\newcommand{\vit}[1]{{{\color{blue} \textbf{Vitaly:} {#1}}}}

\newcommand{\code}{{\mathcal{C}}}

\newcommand{\ff}{{\mathbb{F}}}

\newcommand{\rr}{{\mathbb R}}

\newcommand{\zz}{{\mathbb Z}}

\newcommand{\nn}{{\mathbb N}}

\newcommand{\bldalpha}{{\mbox{\boldmath $\alpha$}}}

\newcommand{\bldg}{{\mbox{\boldmath $g$}}}

\newcommand{\bldG}{{\mbox{\boldmath $G$}}}

\newcommand{\bldx}{{\mbox{\boldmath $x$}}}

\newcommand{\bldy}{{\mbox{\boldmath $y$}}}

\newcommand{\bldzero}{{\mbox{\boldmath $0$}}}

\newcommand{\Tm}[1]{Theorem~\protect\ref{#1}}
\newcommand{\Le}[1]{Lemma~\protect\ref{#1}}

\newcommand{\Ex}[1]{Example~\protect\ref{#1}}

\newcommand{\gth}{\theta}

\newcommand{\beq}{\begin{equation}}
\newcommand{\eeq}{\end{equation}}
\newcommand{\beql}[1]{\begin{equation} \label{#1}}
\newcommand{\eeql}{\end{equation}}
\newcommand{\beqa}{\begin{eqnarray*}}
\newcommand{\eeqa}{\end{eqnarray*}}
\newcommand{\beqal}[1]{\begin{eqnarray} \label{#1}}
\newcommand{\eeqal}{\end{eqnarray}}
\newcommand{\beqan}{\begin{eqnarray}}
\newcommand{\eeqan}{\end{eqnarray}}
\newcommand{\bpf}{\begin{proof}}
\newcommand{\epf}{\end{proof}}
\newcommand{\ben}{\begin{enumerate}}
\newcommand{\een}{\end{enumerate}}
\newcommand{\bit}{\begin{itemize}}
\newcommand{\eit}{\end{itemize}}

\newcommand{\bab}{\begin{abstract}}
\newcommand{\eab}{\end{abstract}}
\newcommand{\bke}{\begin{keywords}}
\newcommand{\eke}{\end{keywords}}

\newcommand{\bfg}{\boldsymbol{g}}

\newcommand{\bfi}{\boldsymbol{i}}

\newcommand{\bfx}{\boldsymbol{x}}

\def\mymedskip{\vskip\medskipamount}
\def\mymedbreak{\par \ifdim\lastskip<\medskipamount
  \removelastskip \penalty-100 \mymedskip \fi}
\def\myaftermedspace{\par \ifdim\lastskip<\medskipamount
  \removelastskip \penalty55\mymedskip\fi}
\newcommand{\eop}{{\unskip\nobreak\hfil\penalty50
          \hskip2em\hbox{}\nobreak\hfil$\Box$
          \parfillskip=0pt \finalhyphendemerits=0 \par}}
\newcommand{\btm}[1]{\begin{theorem} \label{#1}}
\newcommand{\etm}{\end{theorem}}
\newcommand{\ble}[1]{\begin{lemma} \label{#1}}
\newcommand{\ele}{\end{lemma}}
\newcommand{\bpn}[1]{\begin{proposition} \label{#1}}
\newcommand{\epn}{\end{proposition}}
\newcommand{\bex}[1]{\begin{example} \label{#1}}
\newcommand{\eex}{\eop\end{example}}
\newcommand{\bde}[1]{\begin{definition} \label{#1}}
\newcommand{\ede}{\end{definition}}
\newcommand{\bco}[1]{\begin{corollary} \label{#1}}
\newcommand{\eco}{\end{corolllary}}
\newcommand{\bexer}[1]{\begin{exercise} \label{#1}}
\newcommand{\eexer}{\end{exercise}}
\newcommand{\bre}[1]{\begin{remark} \label{#1}}
\newcommand{\ere}{\end{remark}}

\title{\bf On the Minimum Length of Functional Batch Codes with Small Recovery Sets}
\usepackage{authblk}
\author[1]{\sc Kristiina Oksner}
\author[2]{\sc Henk D.L. Hollmann}
\author[3]{\\\sc Ago-Erik Riet}
\author[4]{\sc Vitaly Skachek}
\affil[1,2,3,4]{University of Tartu\thanks{\tt kristiina.oksner@gmail.com,\{henk.hollmann, ago-erik.riet, vitaly\}@ut.ee}
\thanks{This work is supported in part by the Estonian Research Council grant PRG2531.
The work of K. Oksner is also supported in part by the Estonian National Culture Foundation via the Monika Oit Foundation Fellowship.}}
\date{}
\begin{document}
\maketitle

%===BEGIN TITLE PAGE

%THE PAPER BEGINS HERE !!!
%==================================================================================================
\begin{abstract}
Batch codes are of potential use for load balancing and private information retrieval 
in distributed data storage systems. Recently, a special case of batch codes, termed functional 
batch codes, was proposed in the literature. In functional batch codes, users may query 
linear combinations of the information symbols, and not only the information symbols themselves, 
as is the case for standard batch codes. In this work, we consider linear functional batch codes with the 
additional property that every query is answered by using only a small number of coded symbols. We
derive bounds on the minimum length of such codes, and evaluate the bounds numerically. 
\end{abstract}

\medskip


\vspace{-1ex}
\section{Introduction}
 
Batch codes were first proposed in \cite{Batch} for load balancing in multi-server distributed data storage systems. They are also potentially useful in applications such as network switches and private information retrieval. Linear batch codes~\cite{linearBatch} are a special case of batch codes, where the encoding procedure is a linear transformation applied to the information symbols over a finite field. 
The properties of batch codes were extensively studied over the years: the reader can refer, for example, to 
\cite{asi2018constructions, fazeli2015pirlowstorageoverhead, lember2024equalrequestsasymptoticallyhardest, riet2022asynchronous, vardy2016constructions, wang2015optimal}.
%\cite{switchCode}. 
Batch codes with a small size of recovery sets were studied in~\cite{batchRes}. 

Functional batch codes, a special case of batch codes, were studied in \cite{hollmann2023simplex, Length2025, nassar2022array, yohananov2022optimalconstructions, bounds}. Unlike regular batch codes, where only information symbols can be queried, in functional batch codes any linear combination of information symbols can be queried and retrieved. In~\cite{bounds}, a lower bound on the minimal length of a functional batch code was derived.

Both linear batch codes and linear functional batch codes can be defined by their generator matrices $\bldG$. The length of the code is, therefore, equal to the number of columns in $\bldG$. In the literature, usually, the number of code symbols used for a single retrieval is unrestricted. However, in practice, this could lead to extremely high traffic in the system. Therefore, it would be beneficial to restrict the number of symbols used in the retrieval of each query. In this work, we investigate such functional batch codes, where the number of symbols used in each retrieval is small, typically bounded from above by a small constant. In particular, 
we investigate the asymptotic behavior of lower bounds on the minimal length of such codes. 

\vspace{-1ex}
\section{Preliminaries and notation}

In this section, we present notations and some results that will be used in the sequel. 
We denote by $\nn$ the set of natural numbers (including zero), 
${\nn}^{+} = \nn \backslash \{ 0 \}$. We also denote by $\zz$ the set of integers, and by $\rr$ the set of real numbers. 
We adopt the notation $[n]$ for $n \in \nn$, $n \ge 1$, to denote the set $\{1,2,\cdots,n\}$.
Throughout the paper, let $\ff_2$ be a field with two elements. 
For a matrix $\bldG$, we denote by $\bldg^{(j)}$ its $j$-th column. Additionally, $\mathbf{e}_{i}$ denotes the unit vector with a one in position $i$ and zeros elsewhere, and  $\bldzero$ denotes the all-zero vector. 
For a set $T$ of vectors over $\ff_2$, we denote by $\langle T \rangle$ its linear span. 

\begin{definition} For $n\in \nn$ and $m\in \zz$ with $0\leq m\leq n$, we define
$(n)_m:=n(n-1)\cdots (n-m+1)$; in particular, $(n)_0=1$ and $(n)_n=n!$.
\end{definition}
\noindent
Numbers $(n)_m$ are often called \emph{falling factorials}. We use the following well-known result.
\ble{LLfall}
Let $m,n \in \nn^{+}$, $n \ge m$. It holds: $(n)_m\leq (n-(m-1)/2)^m\leq n^m$.
\ele
\bpf The first inequality follows from the arithmetic mean-geometric mean inequality; the second one is obvious.
\epf

\begin{definition}
Let $i_1+\ldots+i_m=n$, where $i_j \in \nn$ for $j \in [m]$, and $n \in \nn$. A multinomial coefficient is defined as follows:
\[
\binom{n}{i_1, i_2, \ldots, i_m} := \frac{n!}{i_1! \, i_2! \, \ldots \, i_m!}  \; .
\]
We sometimes write this as ${n\choose\bfi}$, where $\bfi=(i_1, i_2, \ldots, i_m)$.
\end{definition}
\noindent
Note that the multinomial coefficient above in fact counts the number of partitions of $[n]$ into $m$ disjoint (labeled) subsets of~$[n]$ of sizes $i_1, i_2, \ldots, i_m$, respectively, for $\sum_{j=1}^m i_j=n$.

%(or say: the number of ways to partition $n$ distinct objects into $m$ ordered groups of sizes $k_1,k_2,\ldots, k_m$ (where k_1+k_2+\cdots+k_m=n$)

\begin{theorem} 
%\cite{MatousekNesetril2008} 
\label{MultinomialThorem}
For arbitrary numbers $x_1, x_2, \ldots, x_m \in \rr$ and for $n \in {\nn}$, the following
identity holds:
\vspace{-1ex}
\begin{equation*}
    (x_1 +x_2 + \ldots + x_m)^n = \sum_{\substack{i_1 + i_2 + \ldots + i_m = n\\ i_1, i_2, \ldots, i_m \geq 0}} \binom{n}{i_1, i_2, \ldots, i_m} \; x_1^{i_1} x_2^{i_2}\ldots x_m^{i_m}.
\end{equation*}
\end{theorem}

Consider a (finite or infinite) power series $S(x) = s_0 + s_1x + \cdots + s_n x^n + \cdots$. We denote by $[x^n]S(x)$ the coefficient $s_n$ (corresponding to the term $x^n$) in this series.

%%%%%%%%%%%%%%%%%%%%%%%%

Next, we introduce batch codes. Assume that there is a $k$-symbol database, which is to be distributed among $n$ devices. A client would like to retrieve an arbitrary subset (or \textit{batch}) of $t$ symbols by reading the data stored on the devices. If too many reading requests are directed to the same device, that could create a heavy load on the device, thus slowing down the performance of the system. Batch codes could prevent that situation by improved load balancing. In this work, we study optimal trade-offs between the number of devices $n$ and the number of queries $t$. We focus on binary linear batch codes. 

\begin{definition}[\cite{Batch,vardy2016constructions}]
\label{def:batch}
An $[n, k, t]$ \emph{linear batch code} $\code$ over $\ff_2$ encodes a
vector $\bldx \in \ff_2^k$ into a vector $\bldy \in \ff_2^n$, $\bldy = (y_1, y_2, \cdots, y_n)$, as $\bldy = \bldx \cdot \bldG$, such that for any $t$-tuple of queries $(x_{i_1}, x_{i_2}, \cdots, x_{i_t})$, $i_\ell \in [k]$ for $\ell \in [t]$,
there exist mutually disjoint subsets ${R}_{i_1}, {R}_{i_2}, \ldots, {R}_{i_t}$ of $[n]$, such that each of the values of $x_{i_\ell}$, $\ell \in [t]$, can be retrieved only from the symbols of $\bldy$ with indices in ${R}_{i_\ell}$, respectively. The $k \times n$ matrix $\bldG$ over~$\ff_2$ is called the \emph{generator matrix} of $\code$.  
\label{def:batch}
\end{definition}

The rate of $\code$ is defined as $r := k/n$. 
For simplicity, an $[n, k, t]$ linear batch code is sometimes referred to as a $t$-batch code. We remark that each encoded symbol $y_j$, $j \in [n]$, can be written as $y_j = \sum_{i=1}^k g_{i,j} x_i=\bfx\cdot \bfg^{(j)}$, where 
\[
\bldG =\left( \bfg^{(1)} \mid \bfg^{(2)} \mid \cdots \mid \bfg^{(n)} \right) = \left(g_{i,j}\right)_{i \in[k] ,j \in [n]} \; .
\]

The following definition naturally generalizes Definition~\ref{def:batch}.
\begin{definition}
An $[n, k, t]$ linear \emph{functional} batch code $\code$ over $\ff_2$ 
encodes a vector $\bldx \in \ff_2^k$ into a vector $\bldy \in \ff_2^n$ as $\bldy = \bldx \cdot \bldG$, such that
for any $(\bldalpha_1, \bldalpha_2, \cdots, \bldalpha_t)$, $\bldalpha_\ell \in \ff_2^k$ for $\ell \in [t]$, there exist mutually disjoint subsets  $R_1, R_2, \ldots, R_t$ of $[n]$, and for any $\ell \in [t]$,
the value $\bldx \cdot \bldalpha_\ell^T$ can be retrieved only from 
the symbols of $\bldy$ with indices in $R_\ell$, respectively. 
\end{definition}

For $\ell \in [t]$, we denote $\bldalpha_\ell = (\alpha_{1,\ell}, \alpha_{2,\ell}, \cdots, \alpha_{k,\ell})$. We say that the functional query is 
\[
(\bldx \cdot \bldalpha_1^T, \bldx \cdot \bldalpha_2^T, \cdots, \bldx \cdot \bldalpha_t^T) = \left(\sum_{i = 1}^k \alpha_{i,1} x_{i}, \sum_{i = 1}^k \alpha_{i,2} x_{i}, \cdots, \sum_{i = 1}^k \alpha_{i,t} x_{i}\right) \; ,
\] 
or simply $(\bldalpha_1, \bldalpha_2, \cdots, \bldalpha_t)$. 
We assume $\bldalpha_\ell \neq \bldzero$ for all $\ell$, otherwise the retrieval is trivial. 

We note that Definition~\ref{def:batch} corresponds to the case where $\bldalpha_1, \bldalpha_2, \cdots, \bldalpha_t$ are all unit vectors. It is straightforward to see that an $[n, k, t]$ linear functional batch code is in particular an $[n, k, t]$ linear batch code. Since we consider only linear codes in this work, we drop the word ``linear'' in the sequel.

Theorem 1 in \cite{linearBatch} characterizes the recoverability of the batch of queries
in a batch code. Below, we generalize it towards functional batch codes. 

\begin{theorem}\label{thrm:recovery}
Let $\code$ be an $[n, k, t]$ functional batch code. It is possible to retrieve $(\bldalpha_1, \bldalpha_2, \cdots, \bldalpha_t)$ simultaneously if 
and only if there exist $t$ mutually-disjoint subsets $R_1, R_2, \cdots, R_t$ of indices in $[n]$, 
and for each $R_\ell$, $\ell \in [t]$, there exists a linear combination of columns of $\bldG$ indexed by that set, which is equal to the column vector $\bldalpha_\ell^T$.
\end{theorem}
\begin{proof}
$\,$
%\begin{itemize}
%\item 
First, assume that there exist $t$ non-intersecting subsets $R_1, R_2, \cdots, R_t$ of indices in $[n]$, and for each $R_\ell$, $\bldalpha_\ell^T \in \langle \bldg^{(j)}  \rangle_{j \in R_\ell}$. 
For each $\ell \in [t]$, take
\[
\bldalpha_\ell^T = \sum_{j \in R_\ell} \gamma_j \cdot \bldg^{(j)} \; , 
\]
where all $\gamma_j \in \ff_2$. Due to linearity, the encoding of $\bldx = (x_1, x_2, \cdots, x_k)$ can be written as $\bldy = \bldx \cdot \bldG$.
Then, 
\vspace{-1ex}
\begin{equation*}
\bldx \cdot \bldalpha_\ell^T = \bldx \cdot \left( \sum_{j \in R_\ell} \gamma_j \cdot \bldg^{(j)} \right) 
= \sum_{j \in R_\ell} \gamma_j \cdot \left(\bldx \cdot \bldg^{(j)} \right) 
= \sum_{j \in R_\ell} \gamma_j \cdot y_j \; . 
\end{equation*}
Therefore, the value of $\bldx \cdot \bldalpha_\ell^T$ can be obtained by querying only the values of $y_j$ for $j \in R_\ell$. The conclusion follows from the fact that all $R_\ell$ are disjoint. 

%\item
Next, we show the opposite direction of Theorem~\ref{thrm:recovery}. 
Let $R_\ell$, for $\ell \in [t]$, be a set of indices of entries in $\bldy$, which are used to retrieve $\bldx \cdot \bldalpha^T_\ell$. We show that $\bldalpha_\ell^T \in \langle \bldg^{(j)}  \rangle_{j \in R_\ell}$. 

Denote the vector space $W_\ell := \langle \bldg^{(j)} \rangle_{j \in R_\ell}$.
Assume for a contradiction that $\bldalpha_\ell \notin W_\ell$. 
Recall that the dual space of $W_\ell$, denoted by $W_\ell^\perp$, 
consists of all the vectors orthogonal to
any vector in $W_\ell$. Since $\bldalpha_\ell^T \notin (W_\ell^\perp)^\perp = W_\ell$, there exists a vector $\bldx^T \in W_\ell^\perp$, 
which is not orthogonal to $\bldalpha_\ell^T$, i.e. $\bldx \cdot \bldalpha_\ell^T \neq 0$. 
On the other hand, this vector $\bldx^T$ is orthogonal to any vector $\bldg^{(j)}$ for $j \in R_\ell$. 

Consider the encodings of the vectors $\bldx$ and $\bldzero$, $\bldx \cdot \bldG$ and $\bldzero \cdot \bldG$, respectively. 
In both cases, all the coordinates of $\bldy$ indexed by $R_\ell$ are zeros. Therefore, the result of the retrieval of $\bldalpha_\ell$ in both cases is the same, yet $\bldx \cdot \bldalpha_\ell^T \neq \bldx \cdot \bldzero^T$. We obtain a contradiction. 
%\end{itemize}
\end{proof}

\begin{definition}
Let $\boldsymbol{G} = ( \bldg^{(1)} \mid \bldg^{(2)} \mid \cdots \mid \bldg^{(n)})$ be a $k\times n$ matrix over $\ff_2$. A \emph{recovery set} for a (nonzero) vector $\bldalpha \in \ff_2^k\setminus\{\mathbf{0}\}$ is a subset $R \subseteq [n]$ of column indices such that $\bldalpha$ is contained in the span $\langle \bfg^{(j)} \rangle_{j\in R}$ of the columns of~$\boldsymbol{G}$ with indices in~$R$. We say $R$ is a \emph{minimal recovery set} for~$\bldalpha$ if $R$ does not contain a proper subset that is also a recovery set for $\bldalpha$.
\end{definition}

\begin{definition}
Let $\code$ be an $[n,k,t]$ (functional) batch code. We say that $\code$ has \emph{locality} $r \in \nn^{+}$, 
if for any batch $(\bldalpha_1, \bldalpha_2, \cdots, \bldalpha_t)$ there exist 
disjoint recovery sets $R_1, R_2, \cdots, R_t$ such that $|R_\ell| \le r$ for all $\ell$. 
Such a code will also be denoted as an $[n,k,t,r]$ (functional) batch code.
\end{definition}

\begin{example}\label{BatchExample}
Consider the following linear binary code $\code$ whose binary $2 \times 3$ generator matrix is given by
\vspace{-1ex}
\[ \bldG =
{\small 
\begin{pmatrix}
1 & 0 & 1\\
0 & 1 & 1 
\end{pmatrix} \; .
}
\]
The encoding is given by $(y_1, y_2, y_3) = \bldx \bldG = (x_1, x_2, x_1 + x_2)$, where $\bldx = (x_1, x_2) \in \ff_2^2$.

We can retrieve any pair $(\bldalpha_1, \bldalpha_2)$ as shown in Table~\ref{table:example}. Therefore, $\code$ is a $[3,2,2,2]$ functional batch code. 
\begin{table}[h!]
\scriptsize
\centering
\begin{tabular}{|c|c|c|c|}
\hline
$\bldalpha_1, \bldalpha_2$ & $\bldx \cdot \bldalpha_1^T$ & $\bldx \cdot \bldalpha_2^T$ & $R_1, R_2$ \\
\hline
$\begin{pmatrix} 1, 0 \end{pmatrix}, \begin{pmatrix} 1, 0 \end{pmatrix}$ & $x_1 = y_1$ & $x_1 = y_2 + y_3$ &
$\{1\}, \{2,3\}$ \\
\hline
$\begin{pmatrix} 1, 0 \end{pmatrix}, \begin{pmatrix} 0, 1 \end{pmatrix}$ & $x_1 = y_1$ & $x_2 = y_2$ &
$\{1\}, \{2\}$ \\
\hline
$\begin{pmatrix} 1, 0 \end{pmatrix}, \begin{pmatrix} 1, 1 \end{pmatrix}$ & $x_1 = y_1$ & $x_1 + x_2 = y_3$ &
$\{1\}, \{3\}$ \\
\hline
$\begin{pmatrix} 0, 1 \end{pmatrix}, \begin{pmatrix} 1, 0 \end{pmatrix}$ & $x_2 = y_2$ & $x_1 = y_1$ &
$\{2\}, \{1\}$ \\
\hline
$\begin{pmatrix} 0, 1 \end{pmatrix}, \begin{pmatrix} 0, 1 \end{pmatrix}$ & $x_2 = y_1 + y_3$ & $x_2 = y_2$ &
$\{1,3\}, \{2\}$ \\
\hline
$\begin{pmatrix} 0, 1 \end{pmatrix}, \begin{pmatrix} 1, 1 \end{pmatrix}$ & $x_2 = y_2$ & $x_1 + x_2 = y_3$ &
$\{2\}, \{3\}$ \\
\hline
$\begin{pmatrix} 1, 1 \end{pmatrix}, \begin{pmatrix} 1, 0 \end{pmatrix}$ & $x_1 + x_2 = y_3$ & $x_1 = y_1$ &
$\{3\}, \{1\}$ \\
\hline
$\begin{pmatrix} 1, 1 \end{pmatrix}, \begin{pmatrix} 0, 1 \end{pmatrix}$ & $x_1 + x_2 = y_3$ & $x_2 = y_2$ &
$\{3\}, \{2\}$ \\
\hline
$\begin{pmatrix} 1, 1 \end{pmatrix}, \begin{pmatrix} 1, 1 \end{pmatrix}$ & $x_1 + x_2 = y_1 + y_2$ & $x_1 + x_2 = y_3$ &
$\{1,2\}, \{3\}$ \\
\hline
\end{tabular}
\caption{Recovery sets for possible queries $\bldalpha_1, \bldalpha_2$.}
\label{table:example}
\end{table}
\vspace{-1.5ex}
\end{example}
%%%%%%%%%%%%%%%%%%%%%%%

\begin{definition}
    Let $k \in \nn$, $k \ge 1$. A binary $[2^k-1, k]$ \emph{simplex code} is a linear code of length $n = 2^k-1$ and dimension $k$ over $\ff_2$, whose binary $k \times n$ generator matrix $\bldG$ contains all nonzero vectors in $\ff_2^k$ as columns exactly once.
\end{definition}

%\begin{example}
%    For $k=3$, one possible generator matrix of the simplex code is given by: 
%    \[ \bldG = 
%		{\small
%    \begin{pmatrix}
%1 & 0 & 0 & 1 & 1 & 0 & 1\\
%0 & 1 & 0 & 1 & 0 & 1 & 1\\
%0 & 0 & 1 & 0 & 1 & 1 & 1
%\end{pmatrix} \; .
%    }
%\]
%\end{example}
\vspace{-1ex}
It is conjectured in~\cite{hollmann2023simplex} that an $[2^k-1, k]$ simplex code is 
a $[2^k-1, k, 2^{k-1}, 2]$ functional batch code for any $k \ge 1$ (for example, 
for $k=2$ we obtain the $[3,2,2,2]$ functional batch code in Example~\ref{BatchExample}). 
Moreover, it is shown in \cite{lember2024equalrequestsasymptoticallyhardest}
that if $\bldG$ is a generator matrix of a $[2^k-1, k]$ simplex code, then 
the double-simplex code with the generator matrix $[\bldG \mid \bldG]$ is a $[2^{k+1}-2, k, 2^k, 2]$ functional batch code. 

\vspace{-1ex}
\subsection*{Prior results}

In \cite{bounds}, the authors presented lower bounds on the minimum length $n$ of an $[n,k,t]$ functional batch code. 
\begin{theorem}\cite[Theorem 23]{bounds} \label{BoundForRegularBatch} 
Suppose that $\code$ is an $[n,k,t]$ functional batch code over $\ff_2$. Then, $n/k \geq t/\log(t+1) + o(1)$.
\end{theorem}

In our work, we build on some ideas from the proof of Theorem 23 in \cite{bounds}. We recall the main idea of the proof. Assume that the disjoint subsets $R_1, R_2, \ldots, R_t \subseteq [n]$ are used for recovery of $t$ functional queries. With each such set $R_\ell$, we associate a labeling of the indices in $[n]$ by labels $\ell$, $\ell \in [t]$, where index $j$ is assigned a label $\ell$ if $j\in R_\ell$, and it is assigned a special label 0 if index $j$ does not occur in any of the aforementioned recovery sets. The number of such labellings is $(t+1)^n$. We remark that distinct queries have distinct batches of minimal recovery sets, and thus distinct labellings. 

The number of possible ordered batches of queries $(\bldalpha_1, \bldalpha_2, \ldots, \bldalpha_t)$ for $\bldalpha_i \in \ff_2^k \setminus \{ \bldzero \}$, $i \in [t]$, equals $(2^k-1)^t$. Since the number of distinct labellings is bounded from below by the number of batches of queries, we obtain that  
$(t+1)^n \geq (2^k-1)^t$. \qed

We remark that in the proof of Theorem 23 in \cite{bounds}, it was assumed that the queries are distinct, hence the number of ways to choose $t$ queries is $\binom{2^k-1}{t}t!$. In our proof, we allow the same queries to be repeated more than once. In that case, there are $(2^k-1)^t$ different ordered batches of queries. That approach allows for a slightly more accurate estimate on the number of batches of queries, when compared with the proof in \cite{bounds}.

\vspace{-1ex}
\section{Lower bounds on the minimal length}

In this section, we employ an approach based on generating functions in order to obtain lower bounds on the minimum length of $[n,k,t,r]$ functional batch codes. Consider an arbitrary batch of $t$ queries $\bldalpha_1,\bldalpha_2,\cdots, \bldalpha_t$. Denote the corresponding pairwise-disjoint minimal recovery sets by $R_1, R_2, \cdots, R_t \subseteq [n]$. Assume that the entries of $\bldy$ are indexed by the elements of $[n]$. Similarly to the approach in~\cite{bounds}, we assign some label $\ell \in \{0,1,\cdots,t\}$ to each entry index $j$ in $[n]$. Label $\ell \in [t]$ is assigned to entry $j \in [n]$ if the corresponding $y_j$ is used in the recovery set $R_\ell$, $|R_\ell| \le r$. If entry $j$ is not used in any recovery set, then it is assigned the label $0$.   

Let $\theta_{t,r}(n)$ denote the number of distinct labellings of disjoint subsets $R_1, \ldots, R_t$ of~$[n]$ of size at most $r$, or, equivalently, the number of distinct labellings of the numbers $[n]$ with labels from $\{0,1,\ldots, t\}$, using each nonzero label at least once and at most $r$ times. Our derivations of bounds will be based on the following observation in~\cite{bounds}. 
\btm{LTbasic}
If $\code$ is an $[n,k,t,r]$ functional batch code, then we have $\gth_{t,r}(n)\geq (2^k-1)^t$.
\etm
\vspace{-0.5ex} \noindent For the proof of~\Tm{LTbasic}, see the discussion immediately after \Tm{BoundForRegularBatch}.
\medskip

Obviously, we have that
\begin{equation} \label{eq:theta}
\theta_{t,r}(n) = \sum_{\tiny \begin{array}{c} i_0 \in\nn \;, i_1, i_2, \cdots, i_t \in [r] \; : \\ i_0 + i_1 + i_2 + \cdots + i_t = n \end{array}} \binom{n}{i_0, i_1, \ldots, i_t} \;. 
\end{equation}
%We also define
%\begin{equation} \label{eq:f}
%f_{t,r}(n) = \sum_{\tiny\begin{array}{c} i_1, i_2, \cdots, i_t \in [r] \; : \\ i_1 + i_2 + \cdots + i_t = n\end{array}} \binom{n}{i_1, i_2, \ldots, i_t} \; ,
%\end{equation}
%which is the number of distinct labellings of the numbers in $[n]$ with the labels from $[n]$, using each label at least once and at most $r$ times. Consider the exponential generating functions 
Let
\vspace{-1ex}
\[
\Theta_{t,r}(x):=\sum_{n = 1}^{\infty}\frac{\theta_{t,r}(n)}{n!}x^n 
%\qquad \mbox{ and } \qquad
%F_{t,r}(x):=\sum_{n = 1}^{\infty}\frac{f_{n,t,r}}{n!}x^n
\]
denote the exponential generating function
for the sequence $(\theta_{t,r}(n))_{n\in\nn}$.
% and $(f_{n,t,r})_{n=1,2,3,\ldots}$, respectively.
Write
\[
G_r(x) = \frac{1}{1!} x + \frac{1}{2!} x^2 + \cdots + \frac{1}{r!} x^r \; , 
\]
and note that $e^x=\sum_{j=0}^\infty\frac{x^j}{j!}$. By using (\ref{eq:theta}),
%and (\ref{eq:f}), 
we find that
\beql{LEftrexp}
\Theta_{t,r}(x)=e^x \cdot \left(G_r(x)\right)^t \; .
\eeql 
We can use (\ref{LEftrexp}) to derive a useful recursion for the values of $\gth_{t,r}(n)$. Indeed, since 
\[
\Theta_{t,r}(x)=e^xG_r(x)^t=G_r(x)\Theta_{t-1,r}(x)=\left(x+x^2/2+\cdots+x^r/r!\right)\Theta_{t-1,r}(x) \; ,
\]
we have 
\beqal{LEfrec}
\gth_{t,r}(n) \!\!\!\!&{=}&\!\!\!\! n!\cdot [x^n]\, \left( \left(\frac{x}{1!}+\frac{x^2}{2!}+\cdots+\frac{x^r}{r!}\right)\Theta_{t-1,r}(x)\right) \nonumber\\
&{=}&\!\!\!\! n!\left(\frac{\gth_{t-1,r}(n-1)}{(n-1)! \cdot 1!}+\frac{\gth_{t-1,r}(n-2)}{(n-2)! \cdot 2!}+\cdots +\frac{\gth_{t-1,r}(n-r)}{(n-r)! \cdot r!}\right)\nonumber\\
&{=}&\!\!\!\!\! {n\choose 1}\gth_{t-1,r}(n-1)+{n\choose 2}\gth_{t-1,r}(n-2)+\cdots +{n\choose r}\gth_{t-1,r}(n-r).\;\;\;\;\;\;
\eeqal
Note that, in addition, $\gth_{0,r}(n)=1$ for all $n\geq1$. Recursion (\ref{LEfrec}) provides a convenient way to compute a table of the values of~$\gth_{t,r}(n)$.
Moreover, since 
\[
\gth_{t-1,r}(n-1)\geq \gth_{t-1,r}(n-2)\geq \cdots\geq \gth_{t-1,r}(n-r)\; ,
\]
we obtain that for $n \ge 2r-1$:
\beql{LErecineq}
\gth_{t,r}(n)\leq \left( {n\choose 1}+\cdots +{n\choose r}\right)\cdot\gth_{t-1,r}(n-1)\leq r\cdot{n\choose r}\cdot\gth_{t-1,r}(n-1)\;,
\eeql
which can be useful to prove bounds on~$\gth_{t,r}(n)$ by induction.
We now use (\ref{LEftrexp}) to derive alternative expressions for the values of $\theta_{t,r}(n)$ that are easier to analyze.
\bex{LEr=2est}
By using~(\ref{LEftrexp}), for $r=2$ we have:
\begin{multline}
\theta_{t,2}(n) \;=\; n! \cdot \,[x^n]\, \left((x+x^2/2)^t \sum_{j=0}^\infty\frac{x^j}{j!}\right)\;=\;
%\frac{n(n-1)\cdots(n-t+1)}{2^t} \cdot \sum_{i=0}^t 2^{t-i}{t\choose i} (n-t)\cdots(n-t-i+1) \;.
n!\cdot \sum_{i=0}^{\min\{t,n-t\}} \frac{{t\choose i}}{2^{i}(n-t-i)!}\\
\;\le\; \frac{(n)_t}{2^t} \cdot \sum_{i=0}^t{t \choose i}2^{t-i}(n-t)_i
\; \leq \; (n)_t \cdot (n-t+2)^t/2^t \; ,\label{LEft2nexp1} 
%\\&\leq& (n-(t-1)/2)^t(n-t+2)^t/2^t,\label{LEft2nexp1}
\end{multline}
where in the last step we used \Le{LLfall}.
%As a consequence, if $t\geq2$, we can obtain the estimates
%\beql{LEft2nexp1}
%\theta_{t,2}(n)\leq \frac{(n)_t}{2^t}   \sum_{i=0}^t{t \choose i}2^{t-i}%(n-2)^i=\frac{(n)_t}{2^t} n^t\leq [(n-(t-1)/2)n/2]^t\leq n^{2t}/2^t.
%\eeq
%\begin{equation*}
%\theta_{t,2}(n)\leq \frac{n(n-1)\cdots(n-t+1)}{2^t} \cdot \sum_{i=0}^t 2^{t-i}{t\choose i}(n-2)^i\leq \frac{n(n-1)\cdots(n-t+1)n^t}{2^r}\leq \frac{n^{2t}}{2^r} \; .
%\end{equation*}
%It is not immediately clear how to generalize this idea for $r\geq3$ in a straightforward manner. 
\eex
We now derive an expression for~$\gth_{t,r}(n)$ for general $r$, by extending the approach in~\Ex{LEr=2est}.  Below, if $\bfi=(i_1, \ldots, i_r)$, then we write $|\bfi|:=i_1+2i_2+\cdots+ri_r$; note that $t\leq |\bfi|\leq rt$. 
By using (\ref{LEftrexp}), we have
\beqal{LEftrnexp}
\theta_{t,r}(n)
%=n!\cdot [x^n]G_r(x)^te^x
&=&n!\cdot [x^n]\, \left( (x+x^2/2+\cdots+x^r/r!)^t \sum_{j=0}^\infty\frac{x^j}{j!} \right) \nonumber\\
&=&n! \cdot \sum_{\bfi} \frac{{t \choose \bfi}}{(1!)^{i_1}(2!)^{i_2}\cdots(r!)^{i_r}(n-|\bfi|)!}\nonumber\\
&=& (n)_t \cdot \sum_{\bfi} \frac{{t \choose \bfi}(n-t)_{|\bfi|-t}}{(1!)^{i_1}(2!)^{i_2}\cdots(r!)^{i_r}} \; ,
%&&\hfil
\eeqal
where the sum is over all $r$-tuples $\bfi=(i_1, i_2, \ldots, i_r)$ of non-negative integers with $i_1+i_2+\cdots +i_r=t$ and $|\bfi| \le n$. By~\Le{LLfall}, we have $(n-t)_{|\bfi|-t}\leq (n-t)^{|\bfi|-t}$,
thus
\beqal{LEftrnspec} 
\gth_{t,r}(n) 
&\leq&  (n)_t \cdot \sum_{\bfi}{t \choose \bfi}\left(\frac{(n-t)}{1!}\right)^{i_1}\left(\frac{(n-t)^2}{2!}\right)^{i_2}\cdots\left(\frac{(n-t)^r}{r!}\right)^{i_r}(n-t)^{-t}\nonumber\\
&\le& (n)_t 
\left( \frac{(n-t)^1}{1!}+\frac{(n-t)^2}{2!}+\cdots+\frac{(n-t)^r}{r!}\right)^t(n-t)^{-t}\nonumber\\
&=& (n)_t\left( \frac{1}{1!}+\frac{(n-t)}{2!}+\cdots+\frac{(n-t)^{r-1}}{r!}\right)^t\\ \label{LErest1}
&\leq& (n-(t-1)/2)^t(1+(n-t)/2!+\cdots+(n-t)^{r-1}/r!)^t \; ,\label{LErest2}
\eeqal
where in the last step we used \Le{LLfall}. Finally, if $n\geq t+r$, then 
$1\leq (n-t)/2!\leq \cdots \leq (n-t)^{r-1}/r!$, hence
\vspace{-1ex}
\beql{LErest3}\gth_{t,r}(n)\leq \left((n-(t-1)/2) \cdot \frac{(n-t)^{r-1}}{(r-1)!}\right)^t.
\eeql

%%%%%%%%%%%%%%%%%%%%%%%%%%%%%%%%%%%%%%%%%%%%%%%%%%%

We use our results to obtain the following lower bounds on~$n$. 
\begin{theorem}\label{LLmainest} 
Suppose that $\code$ is an~$[n,k,t,r]$ functional batch code, $n \ge t+r$. 
Then, 
\vspace{-1ex}
\begin{equation}
\label{eq:theta-bound}
2^k-1 \le (n-(t-1)/2) (n-t)^{r-1}/(r-1)! \; . 
\end{equation}
\end{theorem}
\bpf
We use \Tm{LTbasic} together with~(\ref{LErest3}) to obtain~(\ref{eq:theta-bound}). 
Note that (\ref{LErest3}) holds for $n\geq t+r$.
\epf
\begin{corollary}
\label{LCnrest}
If an~$[n,k,t,r]$ binary functional batch code exists and $n\geq t+r$, then 
%\frac{(2r-1)t-1}{2r}
\[n\geq t-(t+1)/2r+\left((2^k-1)(r-1)!\right)^{1/r}.\]
\end{corollary}
\bpf
We use the result in~\Tm{LLmainest}. By the arithmetic mean-geometric mean inequality, we have that
\[ 
(n-(t-1)/2) (n-t)^{r-1}\leq \left(n-t+\frac{t+1}{2r}\right)^r \; .
\]
\vspace{-3ex}
\epf
%
%Note that when $r=2$, these bounds reduce to
%\beql{LEn2est}
%2^k-1\leq (n)_t^{1/t}(n-t+2)/2\leq (n-(t-1)/2)(n-(t-2))/2.
%\eeql
%

We can use the recursion (\ref{LErecineq}) to obtain another (similar) bound.
\btm{LTnrest}If an~$[n,k,t,r]$ binary functional batch code exists and if $n \geq \max\{t+1, 2r-1\}$, then 
\[
n\geq (t+r)/2-1+\left((2^k-1)(r-1)!\right)^{1/r}.\]
\etm
\bpf
By using (\ref{LErecineq}) $t$ times, and noting that $\gth_{0,r}(n-t)=1$ if $n\geq t+1$ , we obtain
\[
\gth_{t,r}(n)\leq \frac{r \cdot (n)_r}{r!} \cdot \gth_{t-1,r}(n-1)\leq \cdots\leq 
r^t \cdot (n)_r \cdot (n-1)_r \cdot (n-t+1)_r/(r!)^t \; .
\]
Now, using the arithmetic mean-geometric mean inequality, we find that
\begin{multline*} 
rt\left((n)_r(n-1)_r\cdots(n-t+1)_r\right)^{1/(rt)}
\;\le\; \sum_{i=0}^{t-1}\sum_{j=0}^{r-1}(n-i-j)\\
\;=\; rtn-r\sum_{i=0}^{t-1}i-t\sum_{j=0}^{r-1}j
\;=\; rtn-rt(t-1)/2-tr(r-1)/2 \; ,
\end{multline*}
\vspace{-0.5ex}
hence
\vspace{-0.5ex} 
\[
\gth_{t,r}(n)\leq \frac{(n-(t+r)/2+1)^{rt}}{((r-1)!)^t} \; .
\]
The bound in the theorem follows from \Tm{LTbasic}. 
\epf

%%%%%%%%%%%%%%%%%%%%%%%%%%%%%%%%%%%%%%%%%%%%%%%%%%%
When $r=2$, the inequality~(\ref{LEft2nexp1}) and Theorem~\ref{LTbasic} imply 
\begin{eqnarray}
(2^k-1) & \le & (n)_t^{1/t}(n-t+2)/2 \; \le \; (n-(t-1)/2)(n-(t-2))/2 \nonumber \\
& \le & \frac{1}{2}(n-3t/4+5/4)^2 \; , \label{eq:2}
\end{eqnarray}
where the second transition follows from Lemma~\ref{LLfall}, and the third transition follows from the arithmetic mean-geometric mean inequality.
The following result follows from~(\ref{eq:2}). 
\begin{theorem} 
\label{thrm:tree}
    For any $[n,k,t,2]$ functional batch code, $t \ge 1$, we have
    \begin{equation*}
        n \geq \sqrt{2 \, (2^k-1)} + 3t/4 - 5/4 \; .
    \end{equation*}
\end{theorem}

\vspace{-1ex}
\section{Numerical results}
\label{numerical-results}
In this section, we present numerical results for the bounds derived. 
In Table~\ref{tab:double-simplex}, we compare lower bounds on the minimal length with the upper bound implied by the double-simplex code, for $r=2$. In columns 3-5, we show the lower bound in Theorem~\ref{thrm:tree}, the lower bound implied by the exact computation of $\theta_{t,r}(n)$ and substitution into Theorem~\ref{LTbasic}, and the
upper bound implied by the analysis in~\cite{lember2024equalrequestsasymptoticallyhardest}, respectively. 
ˇ
\begin{table}[h!]
\scriptsize
\centering
\begin{tabular}{|c|c|c|c|c|}
\hline
$k$ & $t$ & Theorem~\ref{thrm:tree} & Exact $\theta_{t,r}(n)$ & 
Construction \cite{lember2024equalrequestsasymptoticallyhardest} \\
\hline
2 & 4 & 5 & 5 & 6 \\
\hline
3 & 8 & 9 & 10 & 14 \\
\hline
4 & 16 & 17 & 19 & 30 \\
\hline
5 & 32 & 31 & 38 & 62 \\
\hline
6 & 64 & 58 & 74 & 126 \\
\hline
7 & 128 & 111 & 146 & 254 \\
\hline
\end{tabular}
\caption{Comparison of lower and upper bounds on minimal $n$ for $r=2$.}
\label{tab:double-simplex}
\end{table}
\vspace{0ex}

Table~\ref{tab:thrm7} shows the values of the lower bound on $n$ in Theorem~\ref{LTnrest} for various values of $k$, $t$, $r$, and compares them to the bound~\cite[Theorem 23]{bounds} (when the $o(1)$ term is ignored). 

\begin{table}[h!]
\scriptsize
\centering
\begin{tabular}{|c|c|c|c|c|c|}
\hline
$k$ & \cite[Theorem 23]{bounds}, $t=2$ & $t=2$, $r=2$ & $t=2$, $r=3$ & $t=3$, $r=3$ & $t=2$, $r=5$   \\
\hline
5 & 7 & 7 & 6 & 6 & - \\
\hline
6 & 8 & 9 & 7 & 8 & - \\
\hline
7 & 9 & 13 & 8 & 9 & 9 \\
\hline
8 & 11 & 17 & 10 & 10 & 10 \\
\hline
9 & 12 & 24 & 12 & 13 & 10 \\
\hline
10 & 13 & 33 & 15 & 15 & 11 \\
\hline
11 & 14 & 47 & 18 & 18 & 12 \\
\hline
12 & 16 & 65 & 22 & 23 & 13 \\
\hline
13 & 17 & 92 & 27 & 28 & 14 \\
\hline
14 & 18 & 129 & 34 & 34 & 16 \\
\hline
15 & 19 & 183 & 42 & 43 & 18 \\
\hline
\end{tabular}
\caption{Lower bounds in \cite[Theorem 23]{bounds} and Theorem~\ref{LTnrest}.}
\label{tab:thrm7}
\end{table}
\vspace{0ex}

\vspace{-1ex}
\section{Conclusions}
In this work, we considered functional batch codes with a restricted size of recovery sets.
By using exponential generating functions, we obtained estimates on the number of labellings of the indices of codeword symbols. Based on those estimates, we derived bounds on the minimal length of a code.

% BibTeX bibliography

\bibliographystyle{ieeetr}
\bibliography{unitartucs-thesis}

%THE END OF THE PAPER
%=================================================================

\end{document}